\documentclass[aps,prl,twocolumn,superscriptaddress,showpacs,altaffillsymbol,nofootinbib]{revtex4-1} 

\usepackage{amsmath, amsthm, amssymb,amsfonts, graphicx, subfigure}
\usepackage[colorlinks,citecolor=blue,linkcolor=blue, urlcolor=blue, anchorcolor=blue]{hyperref}
\usepackage{makecell}
\usepackage{dcolumn}
\usepackage{mathrsfs}
\usepackage{float}
\usepackage{multirow}
\usepackage{braket}
\usepackage{color}
\usepackage{tcolorbox}
\usepackage[ruled]{algorithm2e}
\usepackage{tikz}
\newtheorem{lemma}{Lemma}
\newtheorem{theorem}{Theorem}
\newtheorem{remark}{Remark}
\newtheorem{definition}{Definition}
\newtheorem{example}{Example}

\begin{document}
\title{Universal  approach to deterministic spatial search via alternating quantum walks}
\author{Qingwen Wang}%
\author{Ying Jiang}%
\author{Shiguang Feng}%
\author{Lvzhou Li}%
\email{Email: lilvzh@mail.sysu.edu.cn}
\affiliation{School of Computer Science and Engineering, Sun Yat-sen University, Guangzhou 510006, China}
	
\begin{abstract}
Spatial search is an important problem in quantum computation, which aims to find a marked vertex on a graph. We propose a novel  approach for designing deterministic quantum search algorithms on a variety of graphs via alternating quantum walks. Our approach is universal because it does not require an instance-specific analysis for different graphs. We highlight the flexibility of our approach by proving that for Johnson graphs, rook graphs, complete-square graphs and complete bipartite graphs, our quantum algorithms can find the marked vertex with $100\%$ success probability and achieve quadratic speedups over classical algorithms. This not only gives an alternative succinct way to prove the existing results, but also leads to new interesting findings on more general graphs.
		
\end{abstract}
	
\maketitle
	
\textit{Introduction}.---
 The continuous-time quantum walk (CTQW)  was introduced by Farhi and Gutmann~\cite{RN1}  in 1998, and has been one of the key components in quantum computation ~\cite{Childs2002,Ambainis2003QUANTUMWA,Farhi2008,Childs2010,Reichardt2012,PhysRevA.91.062304,RN3,Dadras2019,Wang2020,Delvecchio2020,Kadian2021,Silva2023}. 
In 2004, Childs and Goldstone~\cite{RN2} presented an algorithmic framework using  CTQW to solve the spatial search problem which aims to find an unknown marked vertex on an underlying graph of specified topology. They showed that the algorithm has $O(\sqrt{N})$ searching time on complete graphs, hypercubes, and $d$-dimensional periodic lattices for $d>4$, where $N$ is the number of vertices in the graph.
In this framework, a Hamiltonian $H$ is constructed using the adjacency matrix of the graph and information about the location of the marked vertex. The algorithm is then to make a quantum system evolve $T$ time from an initial state under the Hamiltonian $H$, where $T$ can be set arbitrarily.
 Since then, many kinds of graphs, e.g., strong regular graphs~\cite{Janmark2014}, complete bipartite graphs~\cite{RN10}, balanced trees~\cite{RN3}, and Johnson graphs~\cite{RN4}, have been studied in this framework. All of these graphs can admit a quadratic quantum speedup by using CTQW. Specially, it is worth mentioning that an exponential algorithmic speedup can be achieved via CTQW for the welded tree problem~\cite{10.1145/780542.780552}.
    
Recently, an interesting framework called \textit{alternating phase-walk} for spatial search was proposed and applied on a variety of graphs \cite{RN5,RN6,RN8}. In the framework, CTQWs and marked-vertex phase shifts are alternately performed.  More specifically, two Hamiltonians are constructed: One uses the Laplacian matrix or adjacency matrix of the graph and the other uses the information of the marked vertex. 
Then the quantum system evolves alternately under the two Hamiltonians, which is similar to the quantum approximate optimization algorithm (QAOA).
Marsh and Wang  \cite{RN5} first utilized alternating phase-walks to design a deterministic quantum algorithm for spatial search on the class of complete identity interdependent networks (CIINs), which achieves quadratic speedups over classical algorithms. 
Note that a CINN is equivalent to an $n\times 2$ rook graph, and the general case of
an $n \times m$ rook graph was studied by Chakraborty et al.~\cite{chakraborty2020optimality} in
the context of the framework proposed by Childs and Goldstone~\cite{RN2}.

 Later,  Marsh and Wang \cite{RN6} presented a method based on alternating phase-walks for designing quantum spatial search algorithms on periodic graphs, and  by applying this method  they obtained quantum search algorithms with quadratic speedups on Johnson graphs $J(n, 2)$, rook graphs and complete-square graphs.
 However, it is worth pointing out that all the quantum search algorithms given in~\cite{RN6} are not deterministic, or in other words, have a certain probability of failure.
In fact, deterministic spatial search algorithms not only mean we can improve the theoretical success
probability to $100\%$ but also imply a kind of perfect state transfer between two vertices on graphs~\cite{godsil2012state, PhysRevA.90.012331}. Hence, this drives Ref.~\cite{RN6} to propose 
the open problem ``another compelling direction for future research is making the algorithm deterministic''.
Inspired by the above work,  Qu et al.~\cite{RN8} proposed a deterministic quantum spatial search on star graphs via the alternating phase-walk framework, which achieves the well-known lower bound of Grover’s search. 
Again,  how to fully characterize the class of
graphs that permit deterministic search was proposed as a topic of future study in~\cite{RN8}. 
 
In this article, we present a novel and universal approach based on alternating phase-walks to design deterministic quantum spatial search algorithms on a variety of graphs. Taking advantage of the approach, we obtain deterministic quantum search algorithms on Johnson graphs $J(n,k)$ for any fixed $k$, rook graphs, complete-square graphs, and complete bipartite graphs $K(N_1, N_2)$, respectively. All of these algorithms can find the marked vertex with 100\% success probability theoretically and achieve quadratic speedups over classical ones. 
 Our approach is universal because it does not require an instance-specific analysis for different graphs. The results obtained in the paper not only subsume the consequences of~\cite{RN5,RN6,RN8} but also generalize to more graphs. On one hand, the authors of~\cite{RN5} only studied a subset of rook graphs, and the star graph considered in~\cite{RN8} is obviously a special case of complete bipartite graphs. Hence, their results can be obtained as direct conclusions of this work.  On the other hand, in contrast to the algorithms with errors in~\cite{RN6}, our algorithms are fully deterministic. Also, we significantly simplify the proof and generalize from Johnson graphs $J(n,2)$ to $J(n,k)$ for any fixed $k$. 

 The main contribution of this paper is that we introduce a more succinct and efficient formalism of alternating phase-walks based on which a universal framework for deterministic quantum spatial search on a large family of graph classes is proposed.
 The contribution made should be wide of interest. 
  This not only solves the open problem stated in~\cite{RN6} ``another compelling direction for future research is making the algorithm deterministic'',  but also may give inspiration for inventing potential techniques that are applicable to more graph classes.
    
	
	
	\textit{Preliminaries}.--- 
For quantum search in an unstructured database, one frequently-used way is to perform alternately the 
two unitary operators 
\begin{align*}
    U_1(\alpha) & =I-(1-e^{-i\alpha})|s\rangle\langle s|,\\
    U_2(\beta) & =I-(1-e^{-i\beta})|m\rangle\langle m|,
\end{align*}
where $\alpha$ and $\beta$ are real numbers, on the initial state $|s\rangle$ to find the marked element $|m\rangle$.
In~\cite{RN9,RN7}, the authors showed that if the value of $|\langle m|s\rangle|$ is known, then we can choose $\alpha=\pi$ and appropriate values for $\beta$ to carry out the search deterministically.
\begin{lemma}[\cite{RN9,RN7}]\label{l1}
	Given two unitary operators $U_1(\pi)$, $U_2(\beta)$, and a positive number $|\langle m|s\rangle|$, where $U_1(\pi) =I-2|s\rangle\langle s|$, $U_2(\beta)=I-(1-e^{-i\beta})|m\rangle\langle m|$, we can find an integer $p \in O(\frac{1}{|\langle m|s\rangle|})$, and real numbers $\gamma$, $\beta_1,\dots,\beta_p$
	such that 
	\[|m\rangle=e^{-i\gamma}\prod_{k=1}^{p}U_1(\pi) U_2(\beta_k)|s\rangle.\]
\end{lemma}

Let $G=(V, E)$ be a graph where $V$ is the vertex set and  $E$ is the edge set.  Denote ${\cal H}=span\{|v\rangle: v\in V\}$. The continuous-time quantum walk on $G$ starts from the initial state $|\psi(0)\rangle\in {\cal H}$ and evolves by the following Schr\"{o}dinger equation:
\begin{equation}\label{e2}
	i\cdot \frac{d\langle v|\psi(t)\rangle}{dt}=\sum_{u\in V}\langle v|H|u\rangle \\ \langle u|\psi(t)\rangle,
\end{equation}
where $|\psi(t)\rangle$ denotes the state at time $t$, $H$ is a Hamiltonian satisfying $\langle v|H|u\rangle=0$ when $v$ and $u$ are not adjacent in $G$. This equation means that at time $t$, the change of the amplitude of $|v\rangle$ is only related to the amplitudes of its adjacent vertices. From \eqref{e2}, we see that the continuous-time quantum walk over a graph $G$ at time $t$ can be defined by the unitary transformation $U=e^{-iHt}$. One choice of $H$ is the Laplacian matrix $L=D-A$, where $D$ is the degree matrix (a diagonal matrix with $D_{jj}=deg(j)$) and $A$ is the adjacency matrix of $G$. The other choice is to let $H = A$. 

In this article, when we mention a graph, it always refers to a simple undirected connected graph, where ``simple'' means the graph has no loops and has no  multiple edges between any two vertices. Below we give some properties of the  Laplacian matrix $L$.
	\begin{lemma}\label{lemma2}
		Let $G$ be a graph with Laplacian matrix  $L=D-A$, where $D$ is the degree matrix and $A$ is the adjacency matrix of $G$. Then we have the follow properties.
		(\romannumeral 1) $0$ is a simple\footnote{An eigenvalue is said to be simple if its   algebraic multiplicity  is 1.} eigenvalue of $L$, and   the  corresponding eigenvector is $r= (1,1,\dots,1)^T$.
		(\romannumeral 2) Given the  spectral decomposition $L=\sum_{i=1}^{N} \lambda_i |\eta_i \rangle \langle \eta_i|$, $\langle v|\eta_i \rangle$ is a real number for each vertex $v$ of $G$ and each $i\in \{1,2,\dots,N\}$. 
	\end{lemma}
		\begin{proof}
			The property (\romannumeral 1) is a direct conclusion in \cite{spectral}.
				Since $D$ and $A$ are both real and symmetric, $L$ is real and symmetric and its eigenvectors must be real vectors. Thus, property (\romannumeral 2) holds.
		\end{proof}
	
\textit{Search framework}.--- 
Here we will present a universal approach that can be used to design deterministic quantum search algorithms on a variety of graphs. 
Let $G=(V,E)$ be a graph with Laplacian matrix $L$ and the marked vertex $m$. Our aim is to find the location of $m$ via performing alternately the continuous-time quantum walk operator $e^{-iLt}$ and the search operator $e^{-i\theta |m\rangle\langle m|}$ on the initial state $|s\rangle$.
The idea behind our approach is to prove that there exist an integer $p$ and real numbers $\gamma$, $\theta_k$, $t_k$ $(k\in \{1, 2, \dots, p\})$ such that the following equation holds:
	\begin{equation}\label{e3}
		|m\rangle=e^{-i\gamma}\prod_{k=1}^{p}e^{-i\theta_k |m\rangle\langle m|}e^{-iLt_k} |s\rangle.\nonumber
	\end{equation}
	
	Let $S$ be a finite set of integers.  $gcd(S)$ denotes the greatest common divisor of all nonzero elements in $S$. If there is no nonzero element in $S$, then we let $gcd(S)=1$.

	\begin{definition}
		Let $M$ be an $N\times N$ Hermitian matrix  with spectral decomposition $M=\sum_{i=1}^{N} \lambda_i |\eta_i \rangle \langle \eta_i|$, where $\lambda_1,\dots,\lambda_N$ are integers, at least one of which  is $0$. 
		
		(\romannumeral 1) Define $\Lambda_0=\{\lambda_1,\dots,\lambda_N\}$. For $k\geq 0$, we recursively define $\Lambda_{k+1}$ and $\overline{\Lambda}_{k+1}$ as follows
		\begin{equation}\label{e4}
			\Lambda_{k+1}=\{\lambda\in\Lambda_ {k} \mid e^{-i\lambda \frac{\pi}{gcd(\Lambda_ k)}}=1\},
			\nonumber
		\end{equation}
		and \begin{equation}\label{e5}
			\overline{\Lambda}_{k+1}=\{ \lambda\in\Lambda_ {k} \mid e^{-i\lambda \frac{\pi}{gcd(\Lambda_ k)}}=-1\}.
			\nonumber
		\end{equation}
		We use $d_M$ to denote the least $k$ such that $\Lambda_k$ contains only 0.  
		
		(\romannumeral 2) Let $|m\rangle=\sum_{i=1}^{N}\alpha_i |\eta_i\rangle$ be a vector where each $\alpha_i$ is a real number such that
		\[
		(\sum_{\lambda_i \in {\Lambda}_k} \alpha_i^2)(\sum_{\lambda_i \in \overline{\Lambda}_k} \alpha_i^2)\neq 0
		\]
		for any $k\in\{1,\dots,d_M\}$. We let $|w_0\rangle=|m\rangle$, and for $k\in\{1,\dots,d_M\}$ define
		\begin{equation}\label{e6}
			|w_k\rangle =\frac{1}{\sqrt{\sum_{\lambda_i \in {\Lambda}_k} \alpha_i^2}} \sum_{\lambda_i \in {\Lambda}_k} \alpha_i |\eta_i\rangle
			\nonumber
		\end{equation}
		and \begin{equation}\label{e7}
			|\overline{w}_k\rangle =\frac{1}{\sqrt{\sum_{\lambda_i \in \overline{\Lambda}_k} \alpha_i^2}} \sum_{\lambda_i \in \overline{\Lambda}_k} \alpha_i |\eta_i\rangle.
			\nonumber
		\end{equation}
		\label{de1}
	\end{definition}

	\begin{example}
		Given the following $6\times 6$ Hermitian matrix
		\begin{equation}
	    \begin{aligned}
		M &=  0|\eta_1 \rangle \langle \eta_1|+1|\eta_2 \rangle \langle \eta_2|+3|\eta_3 \rangle \langle \eta_3|+6|\eta_4 \rangle \langle \eta_4|\\
	&+64|\eta_5 \rangle \langle \eta_5|
	+64|\eta_6 \rangle \langle \eta_6|,
	    \end{aligned}
        \nonumber
		\end{equation}
		the process of computing $d_M$, $\Lambda_0$, $|w_0\rangle$, and $\Lambda_k$, $\overline{\Lambda}_{k}$, $|w_k\rangle$, $|\overline{w}_{k}\rangle$ for $k\in\{1,2,\dots,d_M\}$ with $|m\rangle=\sum_{i=1}^{6}\alpha_i |\eta_i\rangle$ where $\alpha_1,\dots,\alpha_6$ are not zero is shown in Fig.~\ref{fig1}.
		\begin{figure}[H]
			\centering
			\begin{tikzpicture} [sibling distance =120pt,level distance=60pt,thick=1, scale=1, every node/.style={scale=1}]
				\node(A)[rectangle,minimum width =20pt,minimum height =20pt ,draw=blue,align=center] {$\Lambda_0=\{0,1,3,6,64,64\}$\\$|{w}_0\rangle=|m\rangle=\sum_{i=1}^{6}\alpha_i |\eta_i\rangle$}
				child {
					node (B)[rectangle,minimum width =90pt ,minimum height =20pt ,draw=blue,align=center] {$\overline{\Lambda}_1=\{1,3\}$\\$|\overline{w}_1\rangle=\frac{\alpha_2 |\eta_2\rangle+\alpha_3 |\eta_3\rangle}{\sqrt{\alpha_2^2+\alpha_3^2}}$}				
				}
				child {node(C)[rectangle,minimum width =20pt,minimum height =20pt ,draw=blue,align=center]  {$\Lambda_1=\{0,6,64,64\}$\\$|{w}_1\rangle =\frac{\alpha_1 |\eta_1\rangle+\alpha_4 |\eta_4\rangle+\alpha_5 |\eta_5\rangle+\alpha_6 |\eta_6\rangle}{\sqrt{\alpha_1^2+\alpha_4^2+\alpha_5^2+\alpha_6^2}}$}
					child {node(D)[rectangle,minimum width =90pt,minimum height =20pt ,draw=blue,align=center]  {$\overline{\Lambda}_2=\{6\}$\\$|\overline{w}_2\rangle=\frac{\alpha_4 |\eta_4\rangle}{\sqrt{\alpha_4^2}}$}					
					}
					child {node(E)[rectangle,minimum width =20pt,minimum height =20pt ,draw=blue,align=center]{$\Lambda_2=\{0,64,64\}$\\$|{w}_2\rangle =\frac{\alpha_1 |\eta_1\rangle+\alpha_5 |\eta_5\rangle+\alpha_6 |\eta_6\rangle}{\sqrt{\alpha_1^2+\alpha_5^2+\alpha_6^2}}$}
						child {node(F)[rectangle,minimum width =90pt,minimum height =20pt ,draw=blue,align=center]  {$\overline{\Lambda}_3=\{64,64\}$\\$|\overline{w}_3\rangle=\frac{\alpha_5 |\eta_5\rangle+\alpha_6 |\eta_6\rangle}{\sqrt{\alpha_5^2+\alpha_6^2}}$}
						}
						child {node(G)[rectangle,minimum width =20pt,minimum height =20pt ,draw=blue,align=center] {$\Lambda_3=\{0\}$\\$|{w}_3\rangle =\frac{\alpha_1 |\eta_1\rangle}{\sqrt{\alpha_1^2}} $}
						}
						child [missing] {}
					}
					child [missing] {}
				}
				child [missing] {}
				;
				\begin{scope}[nodes = {draw = none}]
					\path (A) -- (B) node [near start, left]  {$\frac{\lambda_i}{gcd(\Lambda_{0}) }$ is odd};
					\path (A) -- (C) node [near start, right]  {$\frac{\lambda_i}{gcd(\Lambda_{0})}$ is even};
					\path (C) -- (D) node [near start, left]  {$\frac{\lambda_i}{gcd(\Lambda_{1})}$ is odd};
					\path (C) -- (E) node [near start, right]  {$\frac{\lambda_i}{gcd(\Lambda_{1})}$ is even};
					\path (E) -- (F) node [near start, left]  {$\frac{\lambda_i}{gcd(\Lambda_{2})}$ is odd};
					\path (E) -- (G) node [near start, right]  {$\frac{\lambda_i}{gcd(\Lambda_{2})}$ is even};
					\begin{scope}[nodes = {below = 20pt}]
						\node            at (G) {$d_M=3$};
					\end{scope}
				\end{scope}
			\end{tikzpicture}
			\caption{The process of computing $d_M$, $\Lambda_0$, $|w_0\rangle$ and $\Lambda_k$, $\overline{\Lambda}_{k}$, $|w_k\rangle$, $|\overline{w}_{k}\rangle$ for $k\in \{1,\dots,d_M\}$. 
			}
			\label{fig1}
		\end{figure}
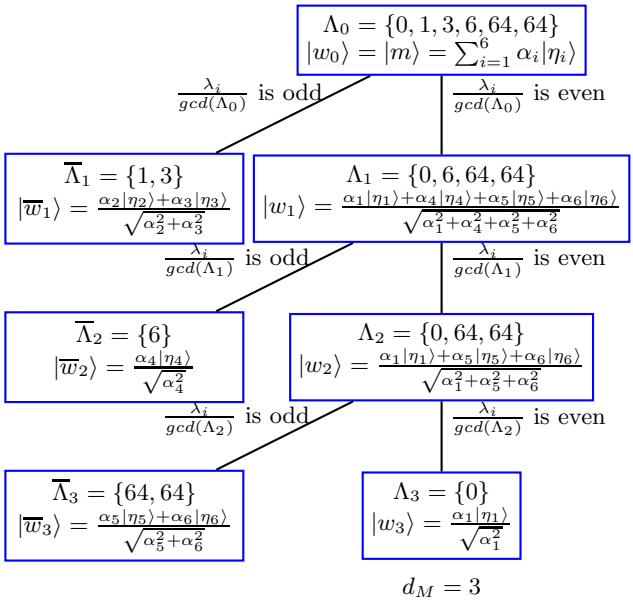
	\end{example}
	
In the following, we consider a  kind of restricted graphs which are	vertex transitive and whose Laplacian matrices have only integer eigenvalues.  A graph $G$ is said to be vertex transitive, if  for any two vertices $v_1$ and $v_2$ of $G$, there is a automorphism $f:G\rightarrow G$ such that $f(v_1)=v_2$. 
Now let $G$ be such a graph with $N$ vertices and  Laplacian matrix $L$. Assume the  spectral decomposition  is $L=\sum_{i=1}^{N} \lambda_i |\eta_i \rangle \langle \eta_i|$.
By Lemma~\ref{lemma2}, one  can see that $L$ satisfies the condition in Definition~\ref{de1}, and thus we can define ${d_L}$, $\Lambda_0$, and $\Lambda_k$, $\overline{\Lambda}_{k}$, $(k\in\{1,2,\dots,{d_L}\})$ as in (\romannumeral 1) of Definition~\ref{de1}.
	In the search space spanned by $\{|v_1\rangle,\dots,|v_N\rangle\}$, $|\eta_1\rangle,\dots,|\eta_N\rangle$ constitute a set of  orthonormal  basis,
	and the marked vertex $|m\rangle$ can be represented as $|m\rangle=\sum_{i=1}^{N}\alpha_i |\eta_i\rangle$, where each $\alpha_i=\langle \eta_i|m\rangle$ is a real number.
	Since $G$ is vertex transitive, as mentioned in \cite{RN6}, we have
	\begin{equation}\label{e14}
		\begin{aligned}
			\sqrt{\sum\nolimits_{\lambda_i \in \Lambda_k} \alpha_i^2}=\sqrt{ \frac{|\Lambda_k|}{N}}, \\
			\sqrt{\sum\nolimits_{\lambda_i \in \overline{\Lambda}_k} \alpha_i^2}=\sqrt{ \frac{|\overline{\Lambda_k}|}{N}},
		\end{aligned}
	\end{equation} 
	where $|\Lambda_k|$ and $|\overline{\Lambda_k}|$  are both positive.\footnote{For a set $S$, $|S|$ denotes its cardinality.}
	For $|m\rangle=\sum_{i=1}^{N}\alpha_i |\eta_i\rangle$, we define $|w_0\rangle$,  $|w_k\rangle$, $|\overline{w}_{k}\rangle$ $(k\in\{1,2,\dots,{d_L}\})$ as in (\romannumeral 2) of Definition~\ref{de1}.
	We can observe that $\Lambda_{k}=\Lambda_{k+1} \cup \overline{\Lambda}_{k+1}$, $|w_k\rangle \in span\{|w_{k+1}\rangle, |\overline{w}_{k+1\rangle}\}$ $(k\in \{0,1,\dots,{d_L}-1\})$ and ${d_L}$ is less than the number of distinct eigenvalues of $L$. 
  
  Below we present one of our main results.

	\begin{theorem}
		Given a vertex transitive graph $G=(V,E)$ with $N$ vertices and Laplacian matrix $L$ such that each eigenvalue of $L$ is an integer, we can find an integer $p \in O(2^{{d_L}-1}\sqrt{N})$, and real numbers $\gamma$, $\theta_k$, $t_k$ $(k\in \{1,2,\dots,p\})$ such that $|s\rangle=e^{-i\gamma}\prod_{k=1}^{p}e^{-i\theta_k |m\rangle\langle m|} e^{-iLt_k}|m\rangle$, where $m$ is the marked vertex and $|s\rangle=\frac{1}{\sqrt{N}}\sum_{v\in V}  |v \rangle $.
		\label{th1}                       
	\end{theorem}
	
	\begin{proof}
		The idea for the proof is to divide the search space into a series of subspaces $span\{|w_{i}\rangle ,|\overline{w}_{i}\rangle \}$ ($i\in\{1,\cdots,d_L\}$). In each subspace $ span\{|w_{i}\rangle ,|\overline{w}_{i}\rangle \}$, we start from $|w_{i-1}\rangle$ and use $e^{-i\theta |m\rangle\langle m|}$ and $e^{-iLt}$ to construct $I- 2|w_i\rangle \langle w_i|$ and $I- (1-e^{-i\theta})|w_{i-1}\rangle \langle w_{i-1}|$ and perform deterministic algorithms to make the state evolve into $|w_i\rangle$ as we do in Lemma \ref{l1}. By repeating the above operation in order, we obtain a procedure which  achieves the following state evolution:
  \begin{equation}
   |m\rangle = |w_0\rangle \rightarrow |w_1\rangle \rightarrow \dots \rightarrow |w_d\rangle=|s\rangle.
  \nonumber
		\end{equation} 
 Therefore, the reversed procedure is what we want.

  We start from the walk operator $e^{-iLt}$. Let $t=\frac{\pi}{gcd(\Lambda_0)}$.  Then
		\begin{equation}\label{e10}
			\begin{aligned}
		e^{-iL  \frac{\pi}{gcd(\Lambda_0)}}&=\sum_{i=1}^{N} e^{-i\lambda_i  \frac{\pi}{gcd(\Lambda_0)}}  |\eta_i \rangle \langle \eta_i|\\
		&=\sum_{\lambda_i \in \Lambda_1} |\eta_i \rangle \langle \eta_i|-\sum_{\lambda_i \in \overline{\Lambda}_1} |\eta_i \rangle \langle \eta_i|.
			\end{aligned}
		\end{equation}
		Recall that \begin{equation}
			|m\rangle=|w_0\rangle \in span\{|w_{1}\rangle ,|\overline{w}_{1}\rangle \}
			\label{e12} 
			\nonumber
		\end{equation}
		and $|w_1\rangle$ ($|\overline{w}_{1}\rangle$) are linear combinations of eigenvectors $|\eta_1\rangle,\dots,|\eta_N\rangle$ whose corresponding eigenvalues are in $\Lambda_1$  ($\overline{\Lambda}_{1}$). We can see that 
  \begin{equation}
  \begin{aligned}
  &e^{-iL  \frac{\pi}{gcd(\Lambda_0)}}|w_1\rangle=|w_1\rangle \\
  &e^{-iL  \frac{\pi}{gcd(\Lambda_0)}}|\overline{w}_{1}\rangle=-|\overline{w}_{1}\rangle.
  \end{aligned}
  \end{equation}
  Restricted to the subspace $span\{|w_{1}\rangle ,|\overline{w}_{1}\rangle \}$, we have
		\begin{equation}\label{e13}
			\begin{aligned}		
				e^{-iL \frac{\pi}{gcd(\Lambda_0)}} & =2|{w}_{1}\rangle \langle {w}_{1}|-I=e^{i\pi}(I-2|{w}_{1}\rangle \langle {w}_{1}|) ,\\
				e^{-i\theta|m\rangle \langle m| } & =I- (1-e^{-i\theta})|w_0\rangle \langle w_0| .
			\end{aligned}
		\end{equation}
		From \eqref{e14}, we have
		\begin{equation}\label{e15}
			\langle w_k|w_{k+1}\rangle=\frac{\sqrt{\sum_{\lambda_i \in \Lambda_{k+1}} \alpha_i^2}}{\sqrt{\sum_{\lambda_i \in \Lambda_k} \alpha_i^2}} =\frac{\sqrt{|\Lambda_{k+1}|} }{\sqrt{|\Lambda_{k}|}},
		\end{equation}
		which is a number independent of $|m\rangle$. According to \eqref{e13}, \eqref{e15} and Lemma~\ref{l1}, we can find parameters $p \in O(\frac{1}{|\langle w_{1}|w_0\rangle|})$, $\gamma$, $t=\frac{\pi}{gcd( \Lambda_0)}$ and $\theta_k$ $(k\in \{1,2,\dots,p\})$ such that
		\begin{equation}
			|w_1\rangle=e^{-i\gamma}\Big(\prod_{k=1}^{p}e^{-i\theta_k |m\rangle\langle m|} e^{-iLt}\Big)|w_0\rangle.
			\label{e16}
		\end{equation}
		Next, we shall prove the following fact.
		
	\begin{lemma}	If there are parameters $p$, $\gamma$, and $\theta_k$, $t_k$ $(k\in \{1,2,\dots,p\})$ satisfying
		\begin{equation}\label{e17}
			|w_i\rangle=e^{-i\gamma}\prod_{k=1}^{p}e^{-i\theta_k |m\rangle\langle m|} e^{-iLt_k}|w_0\rangle,
		\end{equation}
		where $1\leq i \leq {d_L}-1$, then we can find parameters $p'\in O(\frac{2p}{|\langle w_i|w_{i+1}\rangle|})$, $\gamma'$, and $\theta'_k$, $t'_k$, $(k\in \{1,2,\dots, p'\})$ such that
		\begin{equation}	|w_{i+1}\rangle=e^{-i\gamma'}\prod_{k=1}^{p'}e^{-i\theta'_k |m\rangle\langle m|} e^{-iLt'_k}|w_0\rangle.
			\label{e18}
		\end{equation}\label{lm3}
	\end{lemma}
\begin{proof} [Proof of Lemma \ref{lm3}]  Denote $e^{-i\gamma}\prod_{k=1}^{p}e^{-i\theta_k |m\rangle\langle m|} e^{-iLt_k}$ in \eqref{e17} by $A_i$. Then we have
		\begin{equation}\label{e19}
			\begin{aligned}
				A_i e^{-i\theta |m\rangle\langle m|} A_i^{\dagger} &=  A_i(I- (1-e^{-i\theta})|w_0\rangle \langle w_0|) A_i^{\dagger} \\
				& =I- (1-e^{-i\theta})|w_i\rangle \langle w_i|.
				\nonumber
			\end{aligned}
		\end{equation}
		Recall that $ |w_i\rangle \in span\{|w_{i+1}\rangle ,|\overline{w}_{i+1}\rangle \}$. Restricted to this subspace, we have 
		\begin{equation}
			e^{-iL \frac{\pi}{gcd(\Lambda_ {i})}}=2|{w}_{i+1}\rangle \langle {w}_{i+1}|-I=e^{i\pi}(I-2|{w}_{i+1}\rangle \langle {w}_{i+1}|) .
			\label{20}
			\nonumber
		\end{equation}
		Therefore, by Lemma~\ref{l1} there are parameters $ p''\in O(\frac{1}{|\langle w_i|w_{i+1}\rangle|})$, $\gamma''$, $t=\frac{\pi}{gcd(\Lambda_ {i})}$ and $\theta''_k$ $(k\in \{1,2,\dots,p''\})$ such that 
		\begin{equation}\label{e21}
			\begin{aligned}
				|w_{i+1}\rangle & =e^{-i\gamma''}\Big(\prod_{k=1}^{p''} A_i e^{-i\theta''_k |m\rangle\langle m|} A_i^{\dagger}e^{-iLt}\Big)|w_i\rangle \\
				&=e^{-i\gamma''}\Bigr(\prod_{k=1}^{p''}A_i e^{-i\theta''_k |m\rangle\langle m|} A_i^{\dagger}	e^{-iLt}\Bigl) A_i|w_0\rangle.
			\end{aligned}
		\end{equation}
		By replacing $A_i$ with $e^{-i\gamma}\prod_{k=1}^{p}e^{-i\theta_k |m\rangle\langle m|} e^{-iLt_k}$ in \eqref{e21}, we can get the parameters satisfying \eqref{e18}. This proves the fact.
  \end{proof}
  By \eqref{e16} and using Lemma \ref{lm3} recursively, we can find parameters $p \in O(\frac{1}{2}\prod_{k=0}^{{d_L}-1}\frac{2}{|\langle w_{k+1}|w_k\rangle|})$, $\gamma$, and $\theta_k$, $t_k$ $(k\in \{1,2,\dots,p\})$ such that
		\begin{equation} 
			\begin{aligned}
			|w_{{d_L}}\rangle&=e^{-i\gamma}\prod_{k=1}^{p}e^{-i\theta_k |m\rangle\langle m|} e^{-iLt_k}|w_0\rangle\\
			&=e^{-i\gamma}\prod_{k=1}^{p}e^{-i\theta_k |m\rangle\langle m|} e^{-iLt_k}|m\rangle.
			\label{e22}
		\end{aligned}
		 \end{equation}
		By Lemma~\ref{lemma2}, $0$ is a simple eigenvalue of $L$, and $|s\rangle$ is the corresponding eigenvector.
		Thus we have
		\begin{equation}\label{e23} 
			\begin{aligned}
			|w_{{d_L}}\rangle&=\frac{1}{\sqrt{\sum_{\lambda_i \in {\Lambda}_{d_L}} \alpha_i^2}} \sum_{\lambda_i \in {\Lambda}_{d_L}}\alpha_i |\eta_i\rangle\\
			&=\frac{1}{\sqrt{\sum_{\lambda_i =0} \alpha_i^2}} \sum_{\lambda_i =0} \alpha_i |\eta_i\rangle=e^{i\gamma_0}|s\rangle, 
			\nonumber
		\end{aligned}
		\end{equation} where $e^{i\gamma_0}=1$ or $-1$
		and we can ignore it since it is a global phase.
		The upper bound of times that we call the search operator is
		\begin{equation}\label{e24} 
			\begin{aligned}
				\frac{1}{2}\prod_{k=0}^{{d_L}-1}\frac{2}{|\langle w_{k+1}|w_k\rangle|} & =\frac{1}{2}\prod_{k=0}^{{d_L}-1}\frac{2\sqrt{\sum_{\lambda_ \in \Lambda_k} \alpha_i^2}}{\sqrt{\sum_{\lambda_i \in \Lambda_{k+1}} \alpha_i^2}} \\
				& = \frac{2^{{d_L}-1}\sqrt{\sum_{\lambda_ \in \Lambda_0} \alpha_i^2}}{\sqrt{\sum_{\lambda_i \in \Lambda_{{d_L}}} \alpha_i^2}}\\
				& =\frac{2^{{d_L}-1}\sqrt{|\Lambda_{0}|}}{\sqrt{|\Lambda_{{d_L}}}| } \\
				& =2^{{d_L}-1}\sqrt{N}.
				\nonumber
			\end{aligned}
		\end{equation}
		This completes the proof of Theorem \ref{th1} . 
	\end{proof}
\begin{remark}
In Theorem \ref{th1}, 
 Laplacian matrix $L$ can be relaxed to have  rational eigenvalues rather than restricted to integer eigenvalues.
 In fact, for rational numbers $\lambda_i$ $(i\in\{1,2,\dots,N\})$, we can always find a number $q$ such that $q\lambda_i$ $(i\in\{1,2,\dots,N\})$ are all integers.
\end{remark}

\textit{ Applications}.---
It will be shown that by applying Theorem~\ref{th1} we can not only obtain easily the results in \cite{RN5,RN6,RN8},  but also design deterministic quantum spatial search algorithms on some more general graphs.

We first take Johnson graphs as an example. The Johnson graph $J(n, k)$ has vertices given by the $k$-subsets of $\{1,\cdots,n\}$, with two vertices connected when their intersection has size $k-1$. It has many interesting properties and connections with many important problems. Marsh and Wang~\cite{RN6} showed that a quadratic speedup spatial search algorithm on $J(n, 2)$ can be designed using alternating quantum walks. We prove that, for any fixed positive integer $k$, the quadratic speedup can be generalized to the Johnson graphs $J(n, k)$.

\begin{theorem}		\label{th2}
	Let $k$ be a fixed positive integer. For any Johnson graph $J(n, k)$ with Laplacian matrix $L$ and $N$ vertices in which there is a marked vertex $m$, we can design a quantum search algorithm that deterministically finds the marked vertex using $O(\sqrt{N})$ calls to the search operator $e^{-i\theta|m\rangle\langle m|}$.
	
\end{theorem}
\begin{proof}
	The Johnson graph $J(n,k)$ is vertex transitive, and $L$ has $min(k,n-k)+1$ distinct eigenvalues that are all integers \cite{Brouwer2012}.
	By Theorem~\ref{th1}, we can construct a quantum algorithm $A$ satisfying $A|m\rangle=|s\rangle$ that uses $O(2^{{d_L}-1}\sqrt{N})$ calls to the search operator where $|s\rangle$ is the uniform superposition state over all vertices of the graph and ${d_L}$ is defined as in Definition \ref{de1}. Obviously, $A^\dagger|s\rangle=|m\rangle$, which means that the algorithm can find the marked vertex from $|s\rangle$. Recall that ${d_L}$ is less than the number of distinct eigenvalues of $L$ and we have ${d_L}<min(k,n-k)+1\leq k+1$.  Since $k$ is fixed, the number of queries for the algorithm $A^\dagger$ to call the search operator is bounded by $O(\sqrt{N})$. Therefore, the algorithm achieves the quadratic speedup.
\end{proof}
Rook graphs and complete-square graphs are both vertex transitive.
An $m \times n$ rook graph is the graph Cartesian product $K_m \mathbin{\square} K_n$ of complete graphs, having a total of $N = mn$ vertices. The Laplacian matrix of a rook graph has at most four distinct eigenvalues which are all integers \cite{RN6}. 
The $K_n \mathbin{\square} Q_2$ graph that is called complete-square graph is the graph Cartesian product of a complete graph $K_n$ and a square graph $Q_2$. The Laplacian matrix of a complete-square graph has at most six distinct eigenvalues which are all integers \cite{RN6}. Using Theorem~\ref{th1} on rook graphs and complete-square graphs as we do in Theorem~\ref{th2}, we can get similar quantum algorithms that deterministically find the marked vertex with quadratic speedups, which cover the results in~\cite{RN5,RN6}.

Next, we consider the complete bipartite graph $K(N_1, N_2)$, which  usually neither is vertex transitive nor has integer eigenvalues.  So Theorem~\ref{th1} does not apply to this situation directly. However, we can still design a deterministic quantum search algorithm  in a similar way as done in Theorem~\ref{th1}.

A complete bipartite graph $K(N_1, N_2)$ is an undirected graph that has its vertex set partitioned into two subsets $V_1$ of size $N_1$ and $V_2$ of size $N_2$, such that there is an edge from every vertex in $V_1$ to every vertex in $V_2$. A special case of complete bipartite graphs is star graphs, where there is only one vertex in $V_1$ or $V_2$. For this case, Qu et\,al.~\cite{RN8} gave a quantum search algorithm using alternating quantum walks that has $O(\sqrt{N})$ calls to the search operators.  We will give a generalized algorithm for all complete bipartite graphs.

\begin{theorem}
	For any complete bipartite graph $K(N_1, N_2)$ with a marked vertex $m$, we can design a quantum search algorithm that deterministically finds the marked vertex using  $O(\sqrt{N_1+N_2})$ calls to the search operator $e^{-i\theta|m\rangle\langle m|}$.
\end{theorem}
\begin{proof}
	For a complete bipartite graph $K(N_1, N_2)$, its adjacency matrix $A$ has three distinct eigenvalues: 0, $\sqrt{N_1N_2}$, $-\sqrt{N_1N_2}$. The algebraic multiplicity of 0 is $N_1+N_2-2$, and both $\sqrt{N_1N_2}$ and $-\sqrt{N_1N_2}$ are simple eigenvalues. Moreover,  $A$ has the following spectral decomposition
	\begin{equation}\label{e25}
		\begin{aligned}
	    A&=\sqrt{N_1N_2}|\eta_+\rangle \langle \eta_+| -\sqrt{N_1N_2}|\eta_-\rangle \langle \eta_-|\\
		&+\sum_{i=1}^{N_1+N_2-2}0|\eta_i\rangle \langle \eta_i|, 
		\nonumber
	\end{aligned}
	\end{equation}
	where 
	\begin{equation} \label{e26}
		\begin{aligned}
			&\ |\eta_+\rangle=\frac{1}{\sqrt{2N_1N_2}}(\underbrace  {\sqrt{N_2},\dots,\sqrt{N_2}}_{N_1},\underbrace{\sqrt{N_1},\dots,\sqrt{N_1}}_{N_2})^T,\\	&\ |\eta_-\rangle=\frac{1}{\sqrt{2N_1N_2}}(\underbrace  {\sqrt{N_2},\dots,\sqrt{N_2}}_{N_1},\underbrace{-\sqrt{N_1},\dots,-\sqrt{N_1}}_{N_2})^T.
			\nonumber
		\end{aligned}
	\end{equation}
	 We rewrite the marked vertex on this basis as
	\begin{equation}
		|m\rangle=\sum_{i=1}^{N_1+N_2-2}\alpha_i |\eta_i\rangle+\alpha_+|\eta_+\rangle+\alpha_-|\eta_-\rangle.
		\nonumber
		\label{e27}
	\end{equation}
	If the marked vertex is in $V_1$, then 
 \[\langle\eta_+|m\rangle=\langle\eta_-|m\rangle=\frac{1}{\sqrt{2N_1}},\]
	and
	\begin{equation}
		|m\rangle=\sum_{i=1}^{N_1+N_2-2}\alpha_i |\eta_i\rangle+\frac{1}{\sqrt{2N_1}}|\eta_+\rangle+\frac{1}{\sqrt{2N_1}}|\eta_-\rangle.
		\nonumber
		\label{e28}
	\end{equation}
	Define 
	\begin{equation}
		\begin{aligned}
			&\	|\eta_0\rangle=\frac{1}{\sqrt{\sum_{i=1}^{N_1+N_2-2}\alpha_i^2}}\sum_{i=1}^{N_1+N_2-2}\alpha_i |\eta_i\rangle, \\
			&\	|s\rangle=\frac{1}{\sqrt{2}}(|\eta_+\rangle+|\eta_-\rangle)=\frac{1}{\sqrt{N_1}}(1,1,\cdots,1,0,\cdots,0).
			\nonumber
			\label{e29}
		\end{aligned}
	\end{equation}
	We can see $|m\rangle \in span\{|\eta_0\rangle,|s\rangle\}$, and in this subspace,
	\begin{equation}
		e^{-iA\frac{\pi}{\sqrt{N_1N_2}}}=I-2|s\rangle \langle s|.
		\nonumber
		\label{e30}
	\end{equation}
	Obviously, $\langle m|s\rangle=\frac{1}{\sqrt{N_1}}$. Thus,  by Lemma~\ref{l1} we can find parameters $p \in O(\sqrt{N_1})$, $\gamma$, $\theta_k$ $(k\in \{1,2,\dots,p\})$ such that
	\begin{equation}
		|m\rangle=e^{-i\gamma}\prod_{k=1}^{p}e^{-iA\frac{\pi}{\sqrt{N_1N_2}}}e^{-i\theta_k |m\rangle\langle m|}|s\rangle.
		\nonumber
		\label{e31}
	\end{equation} 
	
 Hence, we get a deterministic quantum search algorithm that uses $O({\sqrt{N_1}})$ calls to the search operator and finds the marked vertex from the uniform superposition state over all vertices in $V_1$. Similarly, if the marked vertex is in $V_2$, we can construct a deterministic quantum search algorithm that finds the marked vertex from the uniform superposition state over all vertices in $V_2$ and has $O({\sqrt{N_2}})$ calls to the search operator. Combining the two algorithms above, we can get a quantum algorithm that has $O({\sqrt{N_1}+\sqrt{N_2}})$ calls to the search operator and finds the marked vertex exactly. Since ${\sqrt{N_1}+\sqrt{N_2}}\leq \sqrt{2N_1+2N_2}$, the algorithm has $O(\sqrt{N_1+N_2})$ calls to the search operator 
	and the theorem is proved completely.
\end{proof}

	\textit{Conclusion and Discussion}.--- 
In this article, we have presented a universal approach that can be used to design deterministic quantum algorithms for spatial search on a variety of graphs based on alternating phase-walks. Using this approach, we have obtained deterministic quantum search algorithms with quadratic speedups on Johnson graphs $J(n,k)$ for any fixed $k$, rook graphs, complete-square graphs, and complete bipartite graphs, which not only cover the previous results obtained in \cite{RN5,RN6,RN8}, but also result in some new and general results. Our algorithms are concise and easy to understand, which are simply to perform alternately the continuous-time quantum walk operator $e^{-iHt}$ and the search operator $e^{-i\theta |m\rangle\langle m|}$ on the initial state $|s\rangle$.

For future work, we will consider generalizing our approach to more graphs and designing algorithms for the case of multiple marked vertices.
	
	\bibliographystyle{apsrev4-1}
	\bibliography{references.bib}
	
	

\end{document}